\documentclass[12pt]{extarticle}

\usepackage{amsfonts}
\usepackage{amsmath}
\usepackage{amssymb}
\usepackage{amscd}
\usepackage{amsthm}
\usepackage{mathrsfs}
\usepackage{graphicx}
\usepackage{wasysym}
\usepackage{enumerate}
\usepackage{tikz}
\usepackage{geometry}
\usepackage{physics}
\usepackage{textcomp}
\usepackage[colorlinks=true, urlcolor=blue, linkcolor=blue, citecolor=green, pdfborder={0,0,0}]{hyperref}
\usepackage{inputenc}
\usepackage[explicit]{titlesec}
\usepackage{algorithm}
\usepackage{algpseudocode}
\usepackage{bm}
\usepackage{fullpage}

\newtheorem{theorem}{Theorem}

\newtheorem{lemma}[theorem]{Lemma}

\newtheorem{fact}[theorem]{Fact}

\DeclareFontFamily{U}{mathb}{\hyphenchar\font45}
\DeclareFontShape{U}{mathb}{m}{n}{
<-6> mathb5 <6-7> mathb6 <7-8> mathb7
<8-9> mathb8 <9-10> mathb9
<10-12> mathb10 <12-> mathb12
}{}
\DeclareSymbolFont{mathb}{U}{mathb}{m}{n}
\DeclareMathSymbol{\llcurly}{\mathrel}{mathb}{"CE}
\DeclareMathSymbol{\ggcurly}{\mathrel}{mathb}{"CF}

\newcommand{\mb}{\mathbb}

\newcommand{\mc}{\mathcal}

\newcommand{\tsf}{\textsf}
\newcommand{\mbf}{\mathbf}

\newcommand{\tx}{\text}

\newcommand{\wt}{\widetilde}

\newcommand{\iref}[2]{(\hyperref[#2]{\ref*{#1}.\ref*{#2}})}

\newcommand{\F}{\mathbb{F}}

\title{Decoding Downset codes over a grid}
\author{Srikanth Srinivasan\thanks{Department of Mathematics, IIT Bombay. Email: \texttt{srikanth@math.iitb.ac.in}. Supported by SERB grant MTR/20l7/000958.}
\and Utkarsh Tripathi\thanks{Department of Mathematics, IIT Bombay. Email: \texttt{utkarshtripathi.math@gmail.com}. Supported by the Ph.D. Scholarship of NBHM, DAE, Government of India. }
\and S. Venkitesh\thanks{Department of Mathematics, IIT Bombay. Email: \texttt{venkitesh.mail@gmail.com}. Supported by the Senior Research Fellowship of HRDG, CSIR, Government of India.}
}
\date{\today}

\begin{document}

\maketitle

\begin{abstract}
In a recent paper, Kim and Kopparty (Theory of Computing, 2017) gave a deterministic algorithm for the unique decoding problem for polynomials of bounded total degree over a general grid $S_1\times\cdots \times S_m.$ We show that their algorithm can be adapted to solve the unique decoding problem for the general family of \emph{Downset codes}. Here, a downset code is specified by a family $\mc{D}$ of monomials closed under taking factors: the corresponding code is the space of evaluations of all polynomials that can be written as linear combinations of monomials from $\mc{D}.$
\end{abstract}

Polynomial-based codes play an important role in Theoretical Computer Science in general and Computational Complexity in particular. Combinatorial and computational characteristics of such codes are crucial in proving many of the landmark results of the area, including those related to interactive proofs~\cite{LFKN,Shamir,BFL}, hardness of approximation~\cite{ALMSS}, trading hardness for randomness~\cite{BFNW,STV} etc..

Often, in these applications, we consider polynomials of total degree at most  $d$ evaluated at all points of a finite grid $S = S_1\times \cdots \times S_m\subseteq \mathbb{F}^m$ for some field $\mathbb{F}.$ When $d < k := \min_i\{|S_i|\ |\ i\in [m]\}$, this space of polynomials forms a code of positive distance $\mu := |S|\cdot(1- (d/k))$ given by the well-known DeMillo-Lipton-Schwartz-Zippel lemma~\cite{DL,S,Z} (DLSZ lemma from here on). 

A natural algorithmic question related to this is the \emph{Unique Decoding problem}: given $f:S\rightarrow \mathbb{F}$ that is guaranteed to have (Hamming) distance less than $\mu/2$ from some element $P$ of the code, can we find this $P$ efficiently? This problem was solved in full generality only very recently, by an elegant result of Kim and Kopparty~\cite{kimkopparty2017} who gave a deterministic polynomial-time algorithm for this problem. We refer to this algorithm as the KK algorithm. 

What about the Unique decoding problem when $d \geq k$? In this setting, one must be careful in defining the problem since the space of polynomials of total degree at most $d$ no longer has positive distance.\footnote{e.g. if $d\geq |S_j|$, the non-zero polynomial $\prod_{a\in S_j}(X_j-a)$ of degree at most $d$ vanishes over all of $S$.} However, one can ensure positive distance by enforcing additional \emph{individual degree} constraints: specifically, we require that the degree of each variable $X_i$ in the underlying polynomial be strictly smaller than $|S_i|.$ With a little bit of effort, one can check that the KK algorithm can also be adapted to this setting. In particular, this means that the KK algorithm generalizes a result of Reed~\cite{Reed} from the 1950s, which gives an algorithm for decoding multilinear (i.e. all individual degrees are at most $1$) polynomials over $\{0,1\}^m.$

Motivated by this, we try to understand the scope of applicability of the KK algorithm. Can we show that the KK algorithm works for a cleanly defined general family of codes encompassing the results above? As a possible answer to this question, we put forward the class of \emph{Downset Codes.} A \emph{Downset} $\mc{D}$ is a finite set of monomials (over the variables $X_1,\ldots,X_m$) that is closed under taking factors. The downset code $\mc{C}(S,\mc{D})$ is the code consisting of all polynomials that can be written as linear combinations of monomials from $\mc{D}.$ The space of polynomials of total degree at most $d$, for example, is clearly a downset code. But one can also add individual degree constraints, weighted degree constraints, bounds on the support-size (i.e. number of variables) in any monomial etc.. Downset codes thus yield a fairly general family of codes.

Furthermore, there is a natural variant of the DLSZ lemma that yields a combinatorial characterization of the minimum distance of any downset code $\mc{C}(S,\mc{D})$.\footnote{In particular, a code $\mc{C}(S,\mc{D})$ has positive distance if and only if all the monomials in $\mc{D}$ have individual degree less than $k_i$ w.r.t. each variable $X_i$.} (The proof of this lemma uses a classical theorem of Macaulay~\cite{Macaulay} that reduces the problem of estimating the size of a finite variety to counting the number of non-leading monomials in the ideal of the variety.) It is thus natural to ask if we can solve the unique decoding problem for such codes. 

Our main result is that the KK algorithm can be suitably adapted to yield a deterministic polynomial-time algorithm that solves the unique decoding problem for any downset code $\mc{C}(S,\mc{D}).$ Furthermore, the algorithm and its proof of correctness are quite clean; in particular, working in the fairly general setting of downset codes leads to  a simple abstract analysis of the algorithm.

\section{Preliminaries}

Throughout fix a field $\F.$  Let $S_1,\ldots,S_m$ be finite non-empty subsets of $\F$ and let $S$ denote the grid $S_1\times S_2\times\cdots \times S_m.$ We use $k_i$ to denote $|S_i|.$ Given functions $f,g:S\rightarrow \mb{F}$, we use $\Delta(f,g)$ to denote the Hamming distance between $f$ and $g$, i.e. the number of points where they differ.

Let $\mc{M}$ denote the set $\{0,\ldots,k_1-1\}\times\cdots \times \{0,\ldots,k_m-1\}$ with the natural partial order $\preceq.$ For each $\alpha\in \mc{M},$ we denote by $\nabla(\alpha)$ the set $\{\beta\in \mc{M}\ |\  \alpha \preceq \beta\}$ and by $\Delta(\alpha)$ the set $\{\beta\in \mc{M}\ |\ \beta \preceq \alpha\}.$  For any $\alpha\in\mc{M}$, we will identify the monomial $\mbf{X}^\alpha:= X_1^{\alpha_1}\cdots X_m^{\alpha_m}$ with $\alpha$ and use the monomial notation and the multi-index notation interchangeably.

The following fact is standard.

\begin{fact}
	\label{fac:fnstopolys}
	Each $f:S\rightarrow \mb{F}$ has a unique representation as a polynomial $P(X_1,\ldots,X_m)$ where the degree of $X_i$ in $P$ is at most $k_i-1$ for each $i\in [m]$. Equivalently, there is a natural one-one correspondence between the space of all functions from $S$ to $\mb{F}$ and $\mc{C}(S,\mc{M}).$
\end{fact}

Given a downset $\mc{D}\subseteq \mc{M},$ we associate with it the linear code $\mc{C}(S,\mc{D})$, called a \emph{downset code}, defined by 
\[
\mc{C}(S,\mc{D}) = \{f:S\rightarrow \F\ |\ \text{$f$ can be represented by a linear combination of monomials from $\mc{D}$}\}.
\]

The following lemma allows us to compute the minimum distance $\mu(S,\mc{D})$ for any downset $\mc{D}\subseteq \mc{M}.$ Recall (see e.g. \cite{CLO}) that a \emph{monomial order} on monomials in $X_1,\ldots,X_m$ is a total ordering $\sqsubseteq$ of the monomials that is a well order and moreover satisfies the following for any $\alpha,\beta,\gamma\in \mathbb{N}^m$: $\mbf{X}^{\alpha}\sqsubseteq \mbf{X}^\beta \Rightarrow \mbf{X}^{\alpha + \gamma} \sqsubseteq \mbf{X}^{\beta + \gamma}.$

\begin{lemma}[Schwartz-Zippel Lemma for $\mc{C}(S,\mc{D})$]
	\begin{enumerate}
    \item Let $f\in \mc{C}(S,\mc{M})$ be arbitrary and let $\bm{X}^{\alpha}$ be the leading monomial of $f$ w.r.t. a monomial order. Then, $|\mathrm{Supp}(f)| \geq |\nabla(\alpha)|.$
    \item For each $\alpha\in \mc{M},$ there is an $f:S\rightarrow \F$ such that $|\mathrm{Supp}(f)| = |\nabla(\alpha)|$ and $f$ can be represented by a linear combination of monomials from $\Delta(\alpha).$ In particular, if $\alpha\in \mc{D},$ then such an $f\in \mc{C}(S,\mc{D}).$
    \end{enumerate}
    Thus, $\mu(S,\mc{D}) = \min_{\alpha\in \mc{D}}|\nabla(\alpha)|.$ In particular, given $\mc{D},$ it can be found in polynomial time.
\end{lemma}

\begin{proof}
Item 1 is an easy consequence of the proof of Macaulay's theorem~\cite{Macaulay} (see also~\cite[Chapter 5.3, Proposition 4]{CLO}). For completeness, we present a short proof here. Given any polynomial $P\in \mb{F}[X_1,\ldots,X_m]$, let $\mathrm{mSupp}(P)$  denote the set of monomials with non-zero coefficient in $P$. For any set of monomials $\mc{M}',$ let $\Delta(\mc{M}')$ denote the set of monomials that divide some monomial in $\mc{M}'.$

For any $i\in [m]$, let $f_i(X_i) = \prod_{a\in S_i}(X_i-a).$ Note that $f_i$ is a univariate polynomial of degree $k_i$ that vanishes on $S$. Given any polynomial $P\in \mb{F}[X_1,\ldots,X_m]$, the remainder $P_i$ obtained upon dividing $P$ by $f_i$ has degree $< k_i$ in the variable $X_i$ and evaluates to the same value as $P$ at points in $S$. Further, each monomial in $\mathrm{mSupp}(P_i)$ divides some monomial in $\mathrm{mSupp}(P)$ (i.e. $\mathrm{mSupp}(P_i)\subseteq \Delta(\mathrm{mSupp}(P))$) . Repeating this process, we eventually obtain a polynomial $\widetilde{P}\in \mc{C}(S,\mc{M}\cap \Delta(\mathrm{mSupp}(P)))$ representing the same function as $P$.

Let $A$ be the subset of points in $\mbf{x}\in S$ where $f(\mbf{x}) = 0.$ To prove item 1 of the lemma, it suffices to show that every $g:A\rightarrow \mb{F}$ can be represented as a polynomial from $\mc{C}(S,\mc{M}\setminus \nabla(\alpha)).$ Standard linear algebra then implies that $|A|\leq |\mc{M}| - |\nabla(\alpha)| = |S|-|\nabla(\alpha)|,$ which proves item 1.

To prove the above, fix any $g:A\rightarrow \mb{F}$. By extending $g$ in an arbitrary way to $S$, we know that $g$ can be represented by some polynomial $Q\in \mc{C}(S,\mc{M}).$ If $\mathrm{mSupp}(Q)$ does not contain any monomial from $\nabla(\alpha),$ then we are done. Otherwise, we choose the largest (w.r.t. $\sqsubseteq$) monomial $\mbf{X}^{\beta}\in \nabla(\alpha)\cap \mathrm{mSupp}(Q)$. Let $a$ be the coefficient of $\mbf{X}^{\beta}$ in $Q$. 

Assume that $f(\mbf{X}) = \mbf{X}^{\alpha} + f_1(\mbf{X})$ where $LM(f_1) \sqsubset \mbf{X}^{\alpha}$. Multiplying by $\mbf{X}^{\beta-\alpha}$, we see that the polynomial $\mbf{X}^{\beta}+\mbf{X}^{\beta-\alpha}f_1(\mbf{X})$ vanishes on $A$. Note that $Q_1 = Q-a\mbf{X}^{\beta-\alpha}f$ is a polynomial such that $\mathrm{mSupp}(Q_1)\not\ni \mbf{X}^{\beta}$ that also represents the function $g$ at points in $A$. Repeating this process, we eventually obtain a polynomial $P$ without any monomials from $\nabla(\mbf{X}^\alpha)$ that represents $g$. The polynomial $\widetilde{P}$ (obtained by dividing by $f_i$ as mentioned above) also represents $g$ and furthermore is an element of $\mc{C}(S,\mc{M}\cap \Delta(\mathrm{mSupp}(P)) \subseteq \mc{C}(S,\mc{M}\setminus \nabla(\alpha)).$ 

For item 2, assume that $S_i = \{a^i_1,\ldots,a^i_{k_i}\}$ for each $i\in [m]$ and consider $f(\mbf{X}) = \prod_{i\in [m]}\prod_{j\leq \alpha_i}(X_i-a^i_j).$ 
\end{proof}

Let $\widetilde{\mc{M}}$ denote $\{0,\ldots,k_1-1\}\times\cdots\times\{0,\ldots,k_{m-1}-1\}$ with its natural partial order. Let $\widetilde{S}$ denote the set $S_1\times \cdots \times S_{m-1}.$

Given a downset $\mc{D}\subseteq \mc{M}$ as above, let $\deg_m(\mc{D}) = \max\{j\in [k_m-1]\ |\ \exists\ \alpha\in \mc{D}\ \text{s.t.}\ \alpha_m = j\}.$ For $i\in \{0,\ldots,\deg_m(\mc{D})\},$ define 
\[
\mc{D}_i = \{ \beta\in \widetilde{\mc{M}}\ |\ (\beta,i)\in \mc{D}\}.
\]

The following observation will be useful.
\begin{lemma}
\label{lem:downset}
Let $\mc{D}\subseteq \mc{M}$ be any downset, and let $d = \deg_m(\mc{D}).$
\begin{enumerate}
    \item For each $i\in\{0,\ldots,d\},$ $\mc{D}_i$ is a downset in $\widetilde{\mc{M}}.$ Further, we have $\mc{D}_0\supseteq \mc{D}_1\supseteq \cdots \supseteq \mc{D}_d.$
    \item For each $i\in\{0,\ldots,d\},$ we have $\mu(S,\mc{D}) \leq \mu(\widetilde{S},\mc{D}_i)\cdot \mu(S_m,\{0,\ldots,i\}).$ 
\end{enumerate}
\end{lemma}

\begin{proof}
	\begin{enumerate}
		\item  Clear from the definition, since $\mc{D}$ is a downset.
		\item  Let $U\in\mc{C}(\wt{S},\mc{D}_i)$ and have weight equal to $\mu(\wt{S},\mc{D}_i)$, and $V\in\mc{C}(S_m,\{0,\ldots,i\})$ and have weight equal to $\mu(S_m,\{0,\ldots,i\})$.  Let $\mbf{X}^\alpha$ be any monomial in $U(\mbf{X})$ and $Y^j$ be any monomial in $V(Y)$.  So $\mbf{X}^\alpha\in\mc{D}_i$ and $j\in\{0,\ldots,i\}$. Then by Item 1, $\mc{D}_i\subseteq\mc{D}_j$ and so $\mbf{X}^\alpha\in\mc{D}_j$, that is, $\mbf{X}^\alpha Y^j\in\mc{D}$.  Thus, $U\cdot V\in\mc{C}(S,\mc{D})$ and hence, has weight at least $\mu(S,\mc{D})$.  So we see that
		\[
		\mu(S,\mc{D})\le \mu(\wt{S},\mc{D}_i)\cdot\mu(S_m,\{0,\ldots,i\}).
		\]
	\end{enumerate}
\end{proof}
(Note that the quantity $\mu(S_m,\{0,\ldots,i\})$ is the distance of the degree-$i$ Reed-Solomon code on the set $S_m.$)

As in the result of Kim and Kopparty, we work with the more general problem of decoding \emph{Weighted functions} (or weighted received word). A weighted function over $S$ is a function $w:S\rightarrow \mb{F}\times [0,1]$ or equivalently a pair $(f,u)$ where $f:S\rightarrow \mb{F}$ and $u:S\rightarrow [0,1].$ Given a weighted function $w = (f,u)$ and a function $g:S\rightarrow \mb{F}$, we define their distance $\Delta(w,g)$ by
\[
\Delta(w,g) = \sum_{\mbf{x}: f(\mbf{x}) = g(\mbf{x})} \frac{u(\mbf{x})}{2} + \sum_{\mbf{x}: f(\mbf{x}) \neq g(\mbf{x})} \left(1-\frac{u(\mbf{x})}{2}\right).
\]
A (unweighted) function $h:S\rightarrow \mb{F}$ is identified with the weighted function $(h,u)$ where $u$ is the identically zero function. Note that with this identification, $\Delta((h,u),g)$ agrees with the standard Hamming distance $\Delta(h,g)$ between $h$ and $g$. In particular, the unique decoding problem for $\mc{C}(S,\mc{D})$ immediately reduces to the problem of finding a codeword of distance less than $\mu(S,\mc{D})/2$ from a given weighted function.

We will also need the following lemma about weighted codewords of a downset code $\mc{C}(S,\mc{D})$. The proof (in a more general setting) can be found in \cite[Lemma 2.1]{kimkopparty2017}. 

\begin{lemma}
\label{lem:triangle}
Assume that $G,H\in \mc{C}(S,\mc{D})$ are distinct.\footnote{In the algorithm, we will only need this for $m=1$, i.e. for the Reed-Solomon code.}  Let $f:S\to \mathbb{F}\times [0,1]$ be any weighted received word. Then, $\Delta(f,G) + \Delta(f,H)\geq \Delta(G,H) \geq \mu(S,\mc{D}).$ In particular, both $G$ and $H$ cannot be at distance strictly less than $\mu(S,\mc{D})/2$ from $f$.
\end{lemma}

\section{The main theorem}

\begin{algorithm}
\caption{\tsf{WeightedDownsetDecoder}: Decoding a downset code over a grid}
\label{algo:wddecoder}

\begin{algorithmic}[1]
	\State  Input: $(S,\mc{D},w)$, where
	\begin{itemize}
		\item  $S=S_1\times\cdots\times S_m$ is a finite grid in $\mb{F}^n$ with $k_i=|S_i|$.  \Comment  We have $\wt{S}=S_1\times\cdots\times S_{m-1}$.
		\item  $\mc{D}\subseteq\mc{M}$ is a downset.
		\item  $w:S\to\mb{F}\times[0,1]$ is a weighted received word.
	\end{itemize}
	
	\If{$m=1$}
	\State\Return\tsf{WeightedRSDecoder}($S,\mc{D},w$).  \Comment  Here $S\subseteq\mb{F}$ and $\mc{D}=\{0,\ldots,d\}$, for some $d$.
	\Else
	\State  Define $r:S\to\mb{F}$ and $u:S\to[0,1]$ by
	\[
	w(\mbf{x},y)=(r(\mbf{x},y),u(\mbf{x},y)),\quad\tx{for all }(\mbf{x},y)\in\wt{S}\times S_m.
	\]
		\For{$i=0,\ldots,d=\deg_m(\mc{D})$}
		\State  Define $w_i:S\to\mb{F}\times[0,1]$ by
		\[
		w_i(\mbf{x},y)=\left(r(\mbf{x},y)-\sum_{j=0}^{i-1}Q_j(\mbf{x})y^{d-j},u(\mbf{x},y)\right),\quad\tx{for all }(\mbf{x},y)\in\wt{S}\times S_m.
		\]
			\For{$\mbf{x}\in\wt{S}$}
			\State  Define $w_{i,\mbf{x}}:S_m\to\mb{F}$ by $w_{i,\mbf{x}}(y)=w(\mbf{x},y)$, for all $(\mbf{x},y)\in\wt{S}\times S_m$.
			\State  Let $G_\mbf{x}(Y)=\tsf{WeightedRSDecoder}(S_m,\{0,\ldots,d-i\},w_{i,\mbf{x}})\in\mb{F}[Y]$.
				\If{$\Delta(w_{i,\mbf{x}},G_\mbf{x})<\mu(S_m,\{0,\ldots,d-i\})/2$}
				\State  $\sigma_\mbf{x}=\tx{coeff}(Y^{d-i},G_\mbf{x}(Y))$.
				\State  $\delta_\mbf{x}=\Delta(w_{i,\mbf{x}},G_\mbf{x})$.
				\Else
				\State  $\sigma_\mbf{x}=0$.
				\State  $\delta_\mbf{x}=\mu(S_m,\{0,\ldots,d-i\})/2$.
				\EndIf
			\EndFor
			\State  Define weighted word $f_i:\wt{S}\to\mb{F}\times[0,1]$ and $\delta_i:\wt{S}\to[0,1]$ by
			\[
			f_i(\mbf{x})=\left(\sigma_\mbf{x},\frac{\delta_\mbf{x}}{\mu(S_m,\{0,\ldots,d-i\})/2}\right)=(\sigma_\mbf{x},\delta_i(\mbf{x})).
			\]
			\State  Let $Q_i(\mbf{X})=\tsf{WeightedDownsetDecoder}(\wt{S},\mc{D}_{d-i},f_i)$.
		\EndFor
		\State\Return$\sum_{i=0}^d Q_i(\mbf{X})Y^{d-i}$.
	\EndIf
\end{algorithmic}

\end{algorithm}

\begin{theorem}
\label{thm:main}
There is a deterministic polynomial time algorithm that given, $S,\mc{D}$ and a weighted codeword $w:S\to \F\times [0,1]$ produces a codeword $C\in \mc{C}(S,\mc{D})$ such that $\Delta(w,C) < \mu(S,\mc{D})/2,$ if one exists. (If no such codeword $C$ exists, the algorithm outputs some arbitrary polynomial.)
\end{theorem}

For the case when $m=1$, a result of Forney~\cite{Forney} yields a deterministic polynomial-time algorithm for this problem. We call this algorithm \tsf{WeightedRSDecoder} and refer the reader to~\cite{Forney} or~\cite{kimkopparty2017} for a description.

\begin{proof}
	The algorithm is specified as \tsf{WeightedDownsetDecoder} below (Algorithm~\ref{algo:wddecoder}). We prove its correctness by induction on $m$.
	
	For the base case $m=1$, we simply use the algorithm \tsf{WeightedRSDecoder} and so there is nothing to prove.
	
	Now we assume the correctness of \tsf{WeightedDownsetDecoder} algorithm for $m-1$ indeterminates.
	
	Let $w:S\to\mb{F}\times[0,1]$ be a received weighted word. Suppose there is a $C\in\mc{C}(S,\mc{D})$ with $\Delta(w,C)<\mu(S,\mc{D})/2$. We can write
	\[
	C(\mbf{X},Y) = \sum_{i=0}^d P_i(\mbf{X})Y^{d-i},
	\]
	where $d=\deg_m(\mc{D})$ as in Algorithm~\ref{algo:wddecoder}.
	
	We show, by induction on $i\in \{0,\ldots,d\}$, that the algorithm correctly decodes the polynomial $P_i(\mbf{X}).$ In other words, for each $i\in \{0,\ldots,d\},$ the polynomial $Q_i(\mbf{X})$ computed by Algorithm~\ref{algo:wddecoder} is the same as the polynomial $P_i(\mbf{X}).$
	
	Fix $i\in\{0,\ldots,d\}$.  Assume that the algorithm has correctly decoded $P_j(\mbf{X})$ for each $j < i$. Let $C_i(\mbf{X},Y)=\sum_{j=i}^d P_j(\mbf{X})Y^{d-j}$. Note that $P_j\in\mc{C}(\wt{S},\mc{D}_{d-j})$, for all $j\in\{0,\ldots,d\}$.
	
	To show that $Q_i(\mbf{X}) = P_i(\mbf{X}),$ it is enough to show that $\Delta(f_i,P_i)<\mu(\wt{S},\mc{D}_{d-i})/2$ where $f_i$ is as computed by the algorithm. Then, the induction hypothesis implies that $Q_i(\mbf{X}) = P_i(\mbf{X})$.
	
	Define $w_i:S\to\mb{F}\times[0,1]$ as in the algorithm by 
	\[
	w_i(\mbf{x},y)=\left(r(\mbf{x},y)-\sum_{j=0}^{i-1}Q_j(\mbf{x})y^{d-j},u(\mbf{x},y)\right)=:(r_i(\mbf{x},y),u_i(\mbf{x},y)),\quad\tx{for all }(\mbf{x},y)\in\wt{S}\times S_m.
	\]
	
	By induction, we know that $Q_j(\mbf{X}) = P_j(\mbf{X})$ for $j < i$. Hence, we observe that
	\begin{align*}
	r_i(\mbf{X},Y)-C_i(\mbf{X},Y)&=\bigg(r_i(\mbf{X},Y)+\sum_{j=0}^{i-1}P_j(\mbf{X})Y^{d-j}\bigg)-\bigg(C_i(\mbf{X},Y)+\sum_{j=0}^{i-1}P_j(\mbf{X})Y^{d-j}\bigg)\\
	&=\bigg(r_i(\mbf{X},Y)+\sum_{j=0}^{i-1}Q_j(\mbf{X})Y^{d-j}\bigg)-C(\mbf{X},Y)\\
	&=r(\mbf{X},Y)-C(\mbf{X},Y).
	\end{align*}
	Hence, $\Delta(w_i,C_i)=\Delta(w,C)<\mu(S,\mc{D})/2$.
	
	We now analyze $\Delta(f_i,P_i).$ Recall that we have
	\[
	\Delta(f_i,P_i) = \sum_{\mbf{x}\in \wt{S}: \sigma_\mbf{x} = P_i(\mbf{x})} \frac{\delta_i(\mbf{x})}{2} + \sum_{\mbf{x}\in \wt{S}: \sigma_\mbf{x} \neq P_i(\mbf{x})}\left(1- \frac{\delta_i(\mbf{x})}{2}\right).
	\]
	
	Fix some $\mbf{x}\in\wt{S}$.  Define $C_{i,\mbf{x}}(Y)=C_i(\mbf{x},Y).$ Also, define
	\[
	\Delta(f_i(\mbf{x}),P_i(\mbf{x}))=\begin{cases}
	\dfrac{\delta_i(\mbf{x})}{2},&\sigma_\mbf{x}=P_i(\mbf{x})\\
	1-\dfrac{\delta_i(\mbf{x})}{2},&\sigma_{\mbf{x}}\ne P_i(\mbf{x})
	\end{cases}
	\]We claim that 
	\begin{equation}
	\label{eq:ubd}
	\Delta(f_i(\mbf{x}),P_i(\mbf{x}))\leq \frac{\Delta(w_{i,\mbf{x}},C_{i,\mbf{x}})}{\mu(S_m,\{0,\ldots,d-i\})}.
	\end{equation}
	We prove (\ref{eq:ubd}) by a case analysis.
	
	\begin{enumerate}[(a)]
	  \item $\delta_i(\mbf{x}) = 1.$
	  
	  This implies that $\Delta(w_{i,\mbf{x}},G_{\mbf{x}}) \geq \mu(S_m,\{0,\ldots,d-i\})/2.$ In particular, this implies that $\Delta(w_{i,\mbf{x}},C_{i,\mbf{x}}) \geq \mu(S_m,\{0,\ldots,d-i\})/2$ since otherwise the Reed-Solomon decoder would have returned $C_{i,\mbf{x}}$ instead of $G_{\mbf{x}}.$ This immediately implies (\ref{eq:ubd}).
	  
	  So from now we will assume that $\delta_i(\mbf{x}) = \Delta(w_{i,\mbf{x}},G_{\mbf{x}})/(\mu(S_m,\{0,\ldots,d-i\})/2) < 1.$ By Lemma~\ref{lem:triangle}, it follows that $\Delta(w_{i,\mbf{x}},C_{i,\mbf{x}})\geq \mu(S_m,\{0,\ldots,d-i\})/2 > \Delta(w_{i,\mbf{x}},G_\mbf{x}).$
		\item  $\delta_i(\mbf{x}) < 1$ and $\sigma_\mbf{x}=P_i(\mbf{x})$.
		
		In this case, we immediately have
		\[
		\Delta(f_i(\mbf{x}),P_i(\mbf{x}))=\frac{\delta_i(\mbf{x})}{2}=\frac{\Delta(w_{i,\mbf{x}},G_\mbf{x})}{\mu(S_m,\{0,\ldots,d-i\})}\le\frac{\Delta(w_{i,\mbf{x}},C_{i,\mbf{x}})}{\mu(S_m,\{0,\ldots,d-i\})}.
		\]
		
		\item  $\delta_i(\mbf{x}) < 1$ and $\sigma_\mbf{x}\ne P_i(\mbf{x})$.

		As in the previous case, we have $\dfrac{\delta_i(\mbf{x})}{2}=\dfrac{\Delta(w_{i,\mbf{x}},G_\mbf{x})}{\mu(S_m,\{0,\ldots,d-i\})}$.  But as $\sigma_{\mbf{x}} \neq P_i(\mbf{x}),$ we have $G_\mbf{x}\ne C_{i,\mbf{x}}$. Thus, by Lemma~\ref{lem:triangle}, it follows that  $\Delta(w_{i,\mbf{x}},C_{i,\mbf{x}})+\Delta(w_{i,\mbf{x}},G_\mbf{x})\ge \mu(S_m,\{0,\ldots,d-i\})$.  Hence
		\[
		\Delta(f_i(\mbf{x}),P_i(\mbf{x}))=1-\frac{\delta_i(\mbf{x})}{2}=\frac{\mu(S_m,\{0,\ldots,d-i\})-\Delta(w_{i,\mbf{x}},G_\mbf{x})}{\mu(S_m,\{0,\ldots,d-i\})}\le\frac{\Delta(w_{i,\mbf{x}},C_{i,\mbf{x}})}{\mu(S_m,\{0,\ldots,d-i\})}.
		\]
		
	\end{enumerate}
	This concludes the proof of (\ref{eq:ubd}). Using (\ref{eq:ubd}), we get
	\begin{align*}
	\Delta(f_i,P_i)&=\sum_{\mbf{x}\in \wt{S}}\Delta(f_i(\mbf{x}),P_i(\mbf{x}))\le\sum_{\mbf{x}\in\wt{S}}\frac{\Delta(w_{i,\mbf{x}},C_{i,\mbf{x}})}{\mu(S_m,\{0,\ldots,d-i\})}\\
	&=\frac{\Delta(w_i,C_i)}{\mu(S_m,\{0,\ldots,d-i\})}=\frac{\Delta(w,C)}{\mu(S_m,\{0,\ldots,d-i\})}\\
	&<\frac{\mu(S,\mc{D})}{2\mu(S_m,\{0,\ldots,d-i\})}\\
	&\le\frac{\mu(\wt{S},\mc{D}_{d-i})}{2},\quad\tx{by Lemma~\ref{lem:downset} Item 2}.
	\end{align*}
	This completes the proof.
\end{proof}

\paragraph*{Acknowledgments.} The authors are grateful to Swastik Kopparty and Madhu Sudan for their helpful comments and encouragement.

\bibliographystyle{alpha}
\bibliography{downsetdecoder}

\end{document}